\documentclass[letterpaper,10pt,conference]{ieeeconf}

\IEEEoverridecommandlockouts
\usepackage[dvipdfmx]{graphicx}
\usepackage{amssymb,mathrsfs,amsmath}
\usepackage{unit}
\usepackage{color}
\usepackage{cite}

\newtheorem{theorem}{Theorem}
\newtheorem{lemma}{Lemma}

\newtheorem{remark}{Remark}

\newtheorem{definition}{Definition}

\def\K{\mathcal{K}}
\def\R{\mathcal{R}}
\def\G{\mathcal{L}}
\def\P{\mathcal{P}}
\def\E{\mathcal{E}}
\def\Id{\mathcal{I}}

\def\setD{\mathscr{D}}
\def\setF{\mathscr{F}}
\def\setL{\mathscr{L}}

\title{%
On Koopman Resolvents and Frequency Response of Nonlinear Systems
}

\author{%
Yoshihiko Susuki${}^1$, Natsuki Katayama${}^1$, Alexandre Mauroy${}^2$, and Igor Mezi\'c${}^3$
\thanks{%
The work was partially supported by JSPS Bilateral Collaborations (Grant No. JPJSBP120242202) and JST Moonshot R\&D Grant Number JPMJMS2284.
}%
\thanks{%
${}^1$Y. Susuki and N. Katayama are with the Department of Electrical Engineering, Kyoto University, Kyoto, Japan. 
{\tt\small \{susuki.yoshihiko.5c@, katayama.natsuki.46n@st.\}kyoto-u.ac.jp} 
}%
\thanks{%
${}^2$A. Mauroy is with the Department of Mathematics and Namur Center for Complex Systems (naXys), University of Namur, Belgium. 
{\tt\small alexandre.mauroy@unamur.be}
}%
\thanks{%
${}^3$I. Mezi\'c is with the Department of Mechanical Engineering, University of California, Santa Barbara, USA. 
{\tt\small mezic@ucsb.edu}
}%
}%

\begin{document}

\maketitle
\thispagestyle{empty}
\pagestyle{empty}

\begin{abstract}
This paper proposes a novel formulation of frequency response for nonlinear systems in the Koopman operator framework. 
This framework is a promising direction for the analysis and synthesis of systems with nonlinear dynamics based on (linear) Koopman operators. 
We show that the frequency response of a nonlinear plant is derived through the Laplace transform of the output of the plant, which is a generalization of the classical approach to LTI plants and is guided by the resolvent theory of Koopman operators. 
The response is a complex-valued function of the driving angular frequency, allowing one to draw the so-called Bode plots, which display the gain and phase characteristics. 
Sufficient conditions for the existence of the frequency response are presented for three classes of dynamics. 
\end{abstract}


\section{Introduction}

The frequency response, which characterizes properties of dynamic systems, is a fundamental concept in control engineering. 
For LTI systems, the frequency response is defined as a complex-valued function of the driving (angular) frequency and is a standard tool for their analysis and synthesis \cite{callier1991linear}. 
For nonlinear systems, the frequency response has been studied from multiple directions. 
Basically, due to the presence of nonlinearity, the phenomenological point of view has been addressed, that is, how to locate a {steady-state periodic function} 
under a periodic excitation, including the harmonic balance method \cite{hayashi1964nonlinear}, describing function \cite{khalil2002nonlinear}, Volterra series \cite{brockett1976volterra}, frequency response function \cite{pavlov2007frequency,nakamura2025}, and so on. 
As far as we have surveyed in the literature, their  treatment is mainly based on the time-domain formulation or the state-space formulation, which is contrary to the classical treatment of LTI systems in the frequency domain.
If a similar frequency response guided by a rigorous theory is established for nonlinear systems, then it will provide principles and methodologies for analyzing and synthesizing the nonlinear systems in a systematic and physically intuitive manner.

In \cite{susuki2021koopmana}, we reported the theory of Laplace-domain analysis of nonlinear systems in the Koopman operator framework. 
The framework is a promising direction of the analysis and synthesis of systems with nonlinear dynamics based on (linear) Koopman operators: see, e.g., \cite{mezic2013analysis,mauroy2020koopman,otto2021koopmana,bevanda2021koopman,brunton2022modern}. 
The Koopman operator is a composition operator based on a (possibly nonlinear) transformation \cite{lasota1994chaos}. 
It can {\color{black}preserve} the complete information of the underlying transformation with a linear (but infinite-dimensional) setting, mirroring the classical approach to LTI systems. 
In \cite{susuki2021koopmana}, we showed that, for a continuous-time nonlinear autonomous system with a scalar-valued output, the Laplace transform of the output is represented by the action of the so-called resolvent (operator) of a Koopman generator, which we termed the \emph{Koopman resolvent}. 
This theoretically shows that the Laplace transformation can {\color{black}extract} 
properties of the nonlinear system through the Koopman resolvent.  
Note that the parallel theory to discrete-time nonlinear autonomous systems, that is, $z$-domain analysis of nonlinear systems, was also reported in \cite{susuki2024koopman}. 

The purpose of this paper is to formulate the frequency response of nonlinear systems using the Koopman resolvent. 
Specifically, we propose a novel formulation of the frequency response based on the Laplace-domain representation of nonlinear systems derived in \cite{susuki2021koopmana}. 
To achieve this, we introduce the skew-product form of a nonlinear plant with a periodic input, which converts this non-autonomous system into an equivalent autonomous system.
We then investigate spectral properties of the Koopman resolvent in the (autonomous) skew-product form to formulate the frequency response. 
The main contributions of this paper are summarized as follows:
\begin{enumerate}
\item A novel formulation of the frequency response guided by the theory of Koopman resolvents is presented. 
It is directly related to the spectral properties of the Koopman resolvent and is a generalization of the classical approach to LTI systems. 
\item The frequency response is a complex-valued function of the driving angular frequency, allowing one to draw the so-called Bode plots, which display the gain and phase characteristics. 
\item Sufficient conditions for the existence of the frequency response are presented for three classes of dynamics. 
\end{enumerate}

The rest of this paper is organized as follows. 
Section~\ref{sec:prel} provides the system model and a brief introduction to the Koopman resolvent. 
Section~\ref{sec:form} develops a novel formulation of the frequency response using the Koopman resolvent. 
Solvable examples of the response are also presented. 
Section~\ref{sec:spectral} discusses conditions for utilizing the developed formulation. 
Section~\ref{sec:outro} concludes this paper with a summary and follow-up studies.

\section{Preliminaries}
\label{sec:prel}

\subsection{System model}
This paper considers the continuous-time nonlinear plant, given in the following state-space realization: 
\begin{equation}
\left.
\begin{aligned}
\frac{\dd x}{\dd t} := \dot{x} &= F(x,u), \quad
x\in\mathbb{R}^d, u\in\mathbb{U}\subseteq\mathbb{R} \\
y &= g(x,u), \quad y\in\mathbb{R}
\end{aligned}
\right\},
\label{eq:plant}
\end{equation}
where $x$ is the state of the system, $u$ is the control input, $F:\mathbb{R}^d\times\mathbb{U}\to\mathbb{R}^d$ is a nonlinear vector-valued function, and $g: \mathbb{R}^d\times\mathbb{U}\to\mathbb{R}$ is a scalar-valued fuction called the \emph{observable}.  
The purpose of this paper is to investigate the frequency response of the nonlinear plant \eqref{eq:plant}.

\subsection{Brief introduction to Koopman resolvents}

Here, we summarize the Koopman resolvent as the main mathematical tool in this paper. 
Consider the nonlinear plant without control, that is, $\dot{x}=F(x,u=0)=:F_0(x)$. 
For the current analysis, we assume that the time-forward solution of $\dot{x}=F_0(x)$, denoted as $x(t)={S}^t(x_0), t\geq 0$, starting at an initial state $x_0\in\mathbb{R}^d$ exists uniquely.
The solution induces the one-parameter family of maps ${S}^t: \mathbb{R}^d\to\mathbb{R}^d$, $t\geq 0$, called the semi-flow $\left({S}^t\right)_{t\geq 0}$. 
In the Koopman operator framework, we choose an initial observable $f$ as a scalar-valued function defined on the state space $\mathbb{R}^d$ and analyze its positive semi-orbit generated by the semi-flow through the \emph{composition}, $\{f\circ{S}^t, t\geq 0\}$. 
The semi-orbit is modeled by a one-parameter family of \emph{Koopman operators} for the nonlinear system $\dot{x}=F_0(x)$. 
\begin{definition}
Consider a (Banach) space $\setF$ of observables $f: \mathbb{R}^d\to\mathbb{C}$.  
The family of Koopman operators $\K^t: \setF\to\setF$ associated with the family of maps ${S}^t:\mathbb{R}^d\to\mathbb{R}^d$, $t\geq 0$ is defined through the composition
\[
\K^tf := f\circ{S}^t, \quad \forall f\in\setF. 
\]
\end{definition}
\begin{remark}
From the above argument, $\left({S}^t\right)_{t\geq 0}$ is a semigroup, which implies that $\left(\K^t\right)_{t\geq 0}$ is a semigroup of Koopman operators, called the Koopman semigroup.  
It is essential to state that although the semi-flow is nonlinear, the Koopman semigroup is linear in the sense that $\K^t$ \emph{linearly} acts on the observable space $\setF$. 
\end{remark}

Here, by virtue of the semigroup theory, we introduce the generator for the Koopman semigroup. 
For a strongly continuous semigroup $\left(\K^t\right)_{t\geq 0}$ (see, e.g., \cite{engel2006oneparameter} for its definition), the limit
\[
\lim_{t\downarrow 0}\frac{\K^tf-f}{t} =: \G f, \quad \forall f\in\setD(\G) 
\]
exists in the strong sense, where the domain $\setD(\G)$ of a linear operator $\G$ is a dense set in $\setF$. 
This linear operator induced here is the infinitesimal generator $\G: \setD(\G)\to\setF$ of the Koopman semigroup, called the \emph{Koopman generator}. 
For the nonlinear system $\dot{x}=F_0(x)$, where $x\in\mathbb{R}^d$, the Koopman generator $\G$ is represented through the gradient operator $\nabla_x$ in $\mathbb{R}^d$ as
\[
(\G f)(x)=F_0(x)\cdot\nabla_x f(x), \quad \forall f\in\setD(\G),
\]
where ``$\cdot$" stands for the standard inner product in $\mathbb{R}^d$. 

Finally, we introduce the central mathematical tool in this paper---the Koopman resolvent. 
To this end, let us define the resolvent set and spectrum of the Koopman generator $\G$ for the nonlinear system $\dot{x}=F_0(x)$.
\begin{definition}
\label{def:resolvent_spectrum}
The resolvent set $\rho(\G)$ of the Koopman generator $\G: \setD(\G)\to\setF$ is the set of complex values $\lambda$ such that the inverse of the linear operator $\lambda\mathcal{I}-\G$, denoted as $(\lambda\mathcal{I}-\G)^{-1}$, exists and is bounded on $\setF$, where $\mathcal{I}$ is the identity operator.
Furthermore, the spectrum $\sigma(\G)$ of $\G: \setD(\G)\to\setF$ is defined as the complement set of $\rho(\G)$: $\sigma(\G):=\mathbb{C}\setminus\rho(\G)$.
\end{definition}
\begin{definition}
Suppose that $\rho(\G)$ is not empty. 
The resolvent (operator) of the Koopman generator $\G: \setD(\G)\to\setF$, which we call the Koopman resolvent $\R (s; \G): \setF\to\setF$, is defined as 
\[
\R (s; \G):=(s\mathcal{I}-\G)^{-1}, \quad \forall s\in\rho(\G).
\]
\end{definition}

The spectrum $\sigma(\G)$ is decomposed into the three disjoint sets: point, continuous, and residual spectra (see, e.g., \cite{curtain1995introduction}).  
A complex value $\lambda\in\sigma(\G)$ is in the point spectrum of $\G$, denoted as $\sigma\sub{p}(\G)$, if $\lambda\mathcal{I}-\G$ is not injective. 
Such $\lambda$ is called the \emph{Koopman eigenvalue} of $\G$, so that there exists a non-zero $\phi_\lambda\in\setD(\G)\subset\setF$, called the associated \emph{Koopman eigenfunction}, satisfying
\[
\G\phi_\lambda=\lambda\phi_\lambda.
\]
A complex value $\lambda\in\sigma(\G)$ is in the continuous spectrum if $\lambda\mathcal{I}-\G$ is injective, its range is dense in $\setF$, but $(\lambda\mathcal{I}-\G)^{-1}$ is unbounded.  
The residual spectrum does not appear in this paper and is therefore not considered.

\section{Formulation of frequency response}
\label{sec:form}

This section states the main contribution of this paper: utilizing the Koopman resolvent to formulate the frequency response of the nonlinear plant \eqref{eq:plant} through the skew-product form. 
The skew-product form is studied in operator theory of dynamical systems: see, e.g., \cite{giannakis2023koopman}. 

\subsection{The skew-product form of the forced nonlinear plant}

With the frequency response in mind, we consider the single-frequency excitation $u(t)=u_0\ee^{\ii\omega t}$ to the nonlinear plant \eqref{eq:skew1} and construct the equivalent skew-product form as follows:
\begin{equation}
\left.
\begin{aligned}
\dot{x} &= F(x,u) \\
\dot{u} &= \ii\omega u \\
y &= g(x,u)
\end{aligned}
\right\}, \quad
\ii:=\sqrt{-1}.
\label{eq:skew1}
\end{equation}
Now, the Koopman generator for the skew-product form \eqref{eq:skew1} under a suitable choice of the observable space $\setF$ is introduced through the gradient operator $\nabla_x$ and the differential operator $\DD/\DD u$ as
\[
\begin{aligned}
\G\sub{forced} :=& F(x,u)\cdot \nabla_x + (\ii\omega u)\frac{\DD}{\DD u},\\
& \makebox[3em]{}\G\sub{forced}: \setD(\G\sub{forced})\to\setF,
\end{aligned}
\]
and the associated Koopman resolvent is also introduced as
\begin{equation}
\R(s; \G\sub{forced}) = \left(s\Id-\G\sub{forced}\right)^{-1},  
\forall s\in\rho(\G\sub{forced}).
\label{eq:resolvent}
\end{equation}
Then, by using the theory presented in \cite{susuki2021koopmana}, the Laplace transform $\hat{y}(s; x_0,u_0)$ of the output $y(t)=g(x(t),u(t))$ $(t\geq 0)$ for \eqref{eq:skew1} with an initial state $(x(0),u(0))=(x_0,u_0)$ is represented by the action of the Koopman resolvent for the skew-product form to the observable $g$ as follows:
\begin{equation}
\hat{y}(s; x_0,u_0) = [\R(s; \G\sub{forced})g](x_0,u_0).
\label{eq:Laplace}
\end{equation}
This Laplace transform is equivalent to the output response of the nonlinear plant \eqref{eq:plant} with the periodic input $u(t)=u_0\ee^{\ii\omega t}$, which is the basis for our development of the frequency response.

\subsection{Main proposal}

First, we can motivate our developments in this paper by using the following fact. 
Assume that the Poincar\'e map of the skew form \eqref{eq:skew1} has a globally attracting fixed point. 
Then, the skew form \eqref{eq:skew1} has a periodic orbit of fundamental period $T=2\pi/\omega$ or $T/n$ for some $n\in \mathbb{N}$. 
This periodic orbit will be associated with a steady-state periodic output of the nonlinear plant \eqref{eq:plant} with the periodic input $u(t)=u_0\ee^{\ii\omega t}$, which we will revisit in Section~\ref{subsec:stable}. 

Thus, we provide our working definitions on the frequency response: the fundamental and harmonic cases (Definition~\ref{def:FR}) and the subharmonic case (Definition~\ref{def:FR_subharmonics}).
\begin{definition}
\label{def:FR}
Assume that the nonlinear plant \eqref{eq:plant} with the input $u(t)=u_0\ee^{\ii\omega t}$ $(u_0\neq 0,\omega>0)$ and the output $y=g(x,u)$ has a {steady-state} 
periodic output $y_n(t)$ with the angular frequency $n\omega$ $(n\in\mathbb{N})$. 
Then, the frequency response of the nonlinear plant \eqref{eq:plant} is defined as the coefficient ${\color{black}H_n(\omega; g, x_0, u_0)}$ in the following representation of $y_n(t)$:
\begin{equation}
y_n(t)={\color{black}H_n(\omega; g, x_0, u_0)}(u_0\ee^{\ii\omega t})^n.
\label{eq:def_H}
\end{equation}
\end{definition}
\begin{definition}
\label{def:FR_subharmonics}
Assume that the nonlinear plant \eqref{eq:plant} with the input $u(t)=u_0\ee^{\ii\omega t}$ $(u_0>0,\omega>0)$ and the output $y=g(x,u)$ has a {steady-state} 
periodic output $y_{1/n}(t)$ with the angular frequency $\omega/n$ $(n\in\mathbb{N})$. 
Then, the frequency response of the nonlinear plant \eqref{eq:plant}  is defined as the coefficient ${\color{black}H_{1/n}(\omega; g, x_0, u_0)}$ in the following representation of $y_{1/n}(t)$:
\begin{equation}
y_{1/n}(t)={\color{black}H_{1/n}(\omega; g, x_0, u_0)}(u_0\ee^{\ii\omega t})^{1/n}.
\label{eq:def_H_subharmonics}
\end{equation}
\end{definition}
\begin{remark}
The notion of {``steady state"} 
employed here requires clarification. 
For example, if the state trajectory $x(t)$ exhibits a stable periodic solution in the sense of Lyapunov (see Section~\ref{sec:spectral}), and the observable $f$ is linear, then the corresponding output $y(t)$ will {approach a steady-state} 
periodic function. 
The concept of pullback attractors for non-autonomous systems \cite{caraballo2016applied} can be utilized to define the steady state rigorously. 
\end{remark}
%
%
\begin{remark}
In this discussion, we restrict our attention to the periodic response of the forced nonlinear system. 
Consequently, other possible behaviors---such as quasi-periodic or chaotic oscillations---are excluded from consideration. 
Moreover, the relations $\color{black}H_n(i\omega; g, x_0, u_0)$ and $\color{black}H_{1/n}(i\omega; g, x_0, u_0)$ may not hold for all $\omega \in (0, \infty)$, due to the occurrence of bifurcation phenomena. 
For example, if the saddle-node bifurcation for the target periodic output occurs at a specific finite value of $\omega$, then the associated frequency response is not defined beyond that value. 
\end{remark}

Here, we analyze the frequency response in terms of spectral properties of the Koopman generator $\G\sub{forced}$ for the skew-product form \eqref{eq:skew1}. 
First, we address the fundamental and harmonic cases, that is, the angular frequencies of $\omega, 2\omega, 3\omega,\ldots$ 
The angular frequency $n\omega$ $(n\in\mathbb{N})$ of the periodic output \eqref{eq:def_H} is related to the point spectrum of $\G\sub{forced}$. 
\begin{lemma}
\label{lemma:point}
If the mapping $(x,u)\mapsto u^n$ $(n\in\mathbb{N})$ belongs to $\setF$, then $\ii n\omega$ belongs to the point spectrum $\sigma\sub{p}(\G\sub{forced})$ of the Koopman generator $\G\sub{forced}$. 
The associated Koopman eigenfunction $\phi_{\ii n\omega}(x,u)$ is given by
\[
\phi_{\ii n\omega}(x,u)=u^n.
\] 
\end{lemma}
\begin{proof}
It is straightforward by checking the definition $\G\sub{forced}\phi_{\ii n\omega}={\ii n\omega}\phi_{\ii n\omega}$.
\end{proof}
Second, we address the subharmonic case, that is, the angular frequencies of $\omega/2,\omega/3,\ldots$. 
The angular frequency $\omega/n$ is also related to the point spectrum of $\G\sub{forced}$.
\begin{lemma}
\label{lemma:point_subharmonics}
If the mapping $(x,u)\mapsto u^{1/n}$ $(n\in\mathbb{N})$ belongs to $\setF$, then $\ii \omega/n$ belongs to the point spectrum $\sigma\sub{p}(\G\sub{forced})$ of the Koopman generator $\G\sub{forced}$. 
The associated Koopman eigenfunction $\phi_{\ii \omega/n}(x,u)$ is given by
\[
\phi_{\ii \omega/n}(x,u)=u^{1/n}.
\] 
\end{lemma}
\begin{proof}
It is straightforward by checking the definition $\G\sub{forced}\phi_{\ii \omega/n}=({\ii \omega/n})\phi_{\ii \omega/n}$.
\end{proof}

Now, we state the main results that the frequency responses ${\color{black}H_n(\omega; g, x_0, u_0)}$ and ${\color{black}H_{1/n}(\omega; g, x_0, u_0)}$ can be characterized with the Laplace transform \eqref{eq:Laplace} of the output $y(t)$ through the Koopman resolvent \eqref{eq:resolvent}.
\begin{theorem}
\label{thm:FR}
Consider the nonlinear plant \eqref{eq:plant} with the input $u(t)=u_0\ee^{\ii\omega t}$ and the output $y=g(x,u)$, and associated Koopman resolvent $\R(s; \G\sub{forced})$ through the skew-product form \eqref{eq:skew1}. 
If $\ii n\omega$ $(n\in\mathbb{N})$ is a pole of order 1 for $\R(s; \G\sub{forced})$, then the frequency response ${\color{black}H_n(\omega; g, x_0, u_0)}$ of the nonlinear plant \eqref{eq:plant} is represented as follows:
\begin{align}
{\color{black}H_n(\omega; g, x_0, u_0)} 
&= u_0^{-n}\oint_{|s-\ii n\omega|=\varepsilon} \hat{y}(s; x_0,u_0)\frac{\dd s}{2\pi\ii}
\notag\\
&\makebox[-15mm]{}= u_0^{-n}\oint_{|s-\ii n\omega|=\varepsilon} [\R(s; \G\sub{forced})g](x_0,u_0)\frac{\dd s}{2\pi\ii},
\label{eq:projection}
\end{align}
where $\varepsilon$ is taken appropriately so that the circle does not contain any eigenvalue except for $\ii n\omega$, or equivalently, 
\begin{align}
{\color{black}H_n(\omega; g, x_0, u_0)} & \notag\\
& \makebox[-15mm]{}= u_0^{-n}
\lim_{s\to\ii n\omega}(s-\ii n\omega)[\R(s; \G\sub{forced})g](x_0,u_0).
\label{eq:residue}
\end{align}
\end{theorem}
\begin{proof}
The integral in the right-hand side of \eqref{eq:projection} is the resolvent-based representation of the eigenprojection $\P_{\ii n\omega}g$ of the observable $g$ onto the eigenspace of $\G\sub{forced}$ spanned by the eigenfunction $\phi_{\ii n\omega}$ (see page 178 of \cite{kato1995perturbation}), given by
\[
(\P_{\ii n\omega}g)(x_0,u_0)=\phi_{\ii n\omega}(x_0,u_0)V_n=u_0^nV_n,
\]
where $V_n\in\mathbb{C}$ is called the Koopman mode \cite{rowley2009spectral}, and we used $\phi_{\ii n\omega}(x,u)=u^n$ from Lemma~\ref{lemma:point}.  
That is, we have
\[
{\color{black}H_n(\omega; g, x_0, u_0)}=V_n, 
\]
implying that the frequency response is the Koopman mode. 
Here, $(\P_{\ii n\omega}g)(x_0,u_0)$ is the residue of the Laplace transform $\hat{y}(s; x_0,u_0)=[\R(s; \G\sub{forced})g](x_0,u_0)$ with respect to the pole $\ii n\omega$.   
Now, by using the expansion of the resolvent operator around the pole $\ii n\omega$ of order 1 (see page 180 of \cite{kato1995perturbation}), 
we have
\[
\hat{y}(s; x_0,u_0)
= \frac{u_0^n{\color{black}H_n(\omega; g, x_0, u_0)}}{s-\ii n\omega}+ (\textrm{the residual term}),
\]
where the residual term does not possess the pole $\ii n\omega$. 
By taking the inverse Laplace transformation of the first term on the right-hand side, we derive \eqref{eq:def_H}. 
In addition, Equation \eqref{eq:residue} follows immediately from the above, which is a standard technique in complex analysis and control engineering for calculating the residue.
\end{proof}
\begin{theorem}
\label{thm:FR_subharmonics}
Consider the nonlinear plant \eqref{eq:plant} with the input $u(t)=u_0\ee^{\ii\omega t}$ and the output $y=g(x,u)$, and  associated Koopman resolvent $\R(s; \G\sub{forced})$ through the skew-product form \eqref{eq:skew1}. 
If $\ii \omega/n$ $(n\in\mathbb{N})$ is a pole of order 1 for $\R(s; \G\sub{forced})$, then the frequency response ${\color{black}H_{1/n}(\omega; g, x_0, u_0)}$ of the nonlinear plant \eqref{eq:plant} is represented as follows:
\[
\begin{aligned}
{\color{black}H_{1/n}(\omega; g, x_0, u_0)} 
&= u_0^{-1/n}\oint_{|s-\ii \omega/n|=\varepsilon} \hat{y}(s; x_0,u_0)\frac{\dd s}{2\pi\ii}
\notag\\
&\makebox[-15mm]{}= u_0^{-1/n}\oint_{|s-\ii \omega/n|=\varepsilon} [\R(s; \G\sub{forced})g](x_0,u_0)\frac{\dd s}{2\pi\ii},
\end{aligned}
\]
where $\varepsilon$ is taken appropriately so that the circle does not contain any eigenvalue except for $\ii \omega/n$, or equivalently, 
\[
\begin{aligned}
{\color{black}H_{1/n}(\omega; g, x_0, u_0)} & \\
&\makebox[-20mm]{}= 
u_0^{-1/n}
\lim_{s\to \ii\omega/n}\left(s-\ii \frac{\omega}{n}\right)[\R(s; \G\sub{forced})g](x_0,u_0).
\end{aligned}
\]
\end{theorem}
\begin{proof}
The proof is the same as for Theorem~\ref{thm:FR}.
\end{proof}

The results of Theorems~\ref{thm:FR} and \ref{thm:FR_subharmonics} are related to the main message of this paper: the Koopman mode provides the frequency response of the nonlinear plant \eqref{eq:plant}.  
The amplitude (gain) and argument (phase) of ${\color{black}H_n(\omega; g, x_0, u_0)}$ and ${\color{black}H_{1/n}(\omega; g, x_0, u_0)}$ are visualized as the Bode plots in the classical control engineering and as the response curves in nonlinear oscillations: see, e.g., 
\cite{hayashi1964nonlinear,susuki2007energybased}. 
By using the frequency-domain formulation, it becomes possible to draw the phase plot. 
In addition, by clarifying the frequency response as the Koopman mode, its numerical estimation becomes possible by leveraging dynamic mode decomposition \cite{kutz2016dynamic}, which is a well-established numerical scheme within the Koopman operator framework. 

%
\begin{remark}
The definition is a generalization of the classical approach to LTI systems, as shown in Sections~\ref{subsubsec:linear} and \ref{subsec:linear}.
\end{remark}

\subsection{Solvable examples}
\label{subsec:examples}

\subsubsection{One-dimensional linear case}
\label{subsubsec:linear}
Consider the one-dimensional (1d) linear case $F(x)=ax+bu$ and $g(x,u)=x$ with the constant parameters $a,b\in\mathbb{R}$. 
The frequency response from $u$ to $y=x$ is known as $b/(\ii\omega -a)$. 
For the skew-product form, the (principal) Koopman eigenvalues are $a$ and $\ii\omega$, and the associated Koopman eigenfunctions are the following: 
\[
\left.
\begin{aligned}
\phi_{a}(x,u) &= -(\ii\omega-a)x+bu \\
\phi_{\ii\omega}(x,u) &=u
\end{aligned}
\right\},
\]
where $\phi_a$ depends on $\omega$. 
Thus, the decomposition of the time evolution $x(t)$ based on the Koopman eigenfunctions, which is called the Koopman mode decomposition \cite{rowley2009spectral,mezic2013analysis}, is derived as follows:
\[
\begin{aligned}
x(t; x_0,u_0)
&= \ee^{at}\phi_a(x_0,u_0)\frac{-1}{\ii\omega-a} \\
& + \ee^{\ii\omega t}\underbrace{\phi_{\ii\omega}(x_0,u_0)}_{u_0}
\frac{b}{\ii\omega-a},
\end{aligned}
\]
where both the Koopman modes $-1/(\ii\omega -a)$ and $b/(\ii\omega-a)$ depend on $\omega$. 
The Laplace transform $\hat{x}(s; x_0,u_0)$ of $x(t)$ is directly the following:
\[
\begin{aligned}
\hat{x}(s; x_0,u_0) 
&= \frac{1}{s-a}\phi_a(x_0,u_0)\frac{-1}{\ii\omega-a} \\
&+ \underbrace{\frac{1}{s-\ii\omega}\underbrace{\phi_{\ii\omega}(x_0,u_0)}_{u_0}}_{\hat{u}(s; u_0)}
\frac{b}{\ii\omega-a}.
\end{aligned}
\]
Theorem~\ref{thm:FR} suggests that the frequency response ${\color{black}H_1(\omega; g=x, x_0, u_0)}$ is $b/(\ii\omega-a)$, which corresponds to the frequency response in the linear theory.  

\subsubsection{Two-dimensional nonlinear case}
Consider the two-dimensional (2d) nonlinear system, given by
\begin{equation}
\left.
\begin{aligned}
\dot{x}_1 &= a_1x_1+x_2^2\\
\dot{x}_2 &= a_2x_2+u
\end{aligned}
\right\},
\label{eq:2d_nonlinear}
\end{equation}
where $a_1,a_2<0$ $(a_1\neq a_2)$ are the parameters . 
The associated skew-product form with $\omega>0$ has the three (principal) Koopman eigenvalues $a_1,a_2,\ii\omega$ and the associated Koopman eigenfunctions as follows:
\begin{equation}
\left.
\begin{aligned}
\phi_{a_1}(x_1,x_2,u) &= x_1+\frac{x_2^2}{a_1-2a_2} \\
& + \frac{2x_2u}{(a_1-2a_2)(a_1-a_2-\ii\omega)} \\
& +\frac{2u^2}{(a_1-2a_2)(a_1-a_2-\ii\omega)(a_1-\ii 2\omega)} \\
& \makebox[12em]{} (a_1\neq 2a_2)\\
\phi_{a_2}(x_1,x_2,u) &= x_2+\frac{u}{a_2-\ii\omega} \\
\phi_{\ii\omega}(x_1,x_2,u) &= u\\
\end{aligned}
\right\}.
\label{eq:2d_nonlinear_KEF}
\end{equation}
For the output $y=x_1$, calculating Koopman modes for $\lambda=\ii n\omega$ (see Appendix~\ref{appe:KMD}), we can derive the frequency responses at the Koopman eigenvalues  $\ii n\omega$ ($n\in\mathbb{N}$) as follows:
\[
\left.
\begin{aligned}
{\color{black}H_n(\omega; g=x_1, x_0, u_0)} 
&= 0 \quad (n=1,3,4,\ldots,)\\
{\color{black}H_2(\omega; g=x_1, x_0, u_0)}
&= \frac{1}{(\ii 2\omega-a_1)(\ii\omega-a_2)^2}
\end{aligned}
\right\}.
\]
This result is relevant because, clearly from the system's nonlinearity, the steady-state oscillation in $x_1(t)$ contains only the $2\omega$ component.  
For $y=x_2$, we can also derive the frequency responses at the angular frequencies $n\omega$ as follows:
\[
\left.
\begin{aligned}
{\color{black}H_1(\omega; g=x_2, x_0, u_0)} 
&= \frac{1}{\ii\omega-a_2}\\
{\color{black}H_n(\omega; g=x_2, x_0, u_0)} 
&= 0 \quad (n=2,3,\ldots)
\end{aligned}
\right\}.
\]
This result is also relevant because the steady-state oscillation in $x_2(t)$ contains only the $\omega$ component, according to the linear system $\dot{x}_2=a_2x_2+u$.  
The non-zero responses ${\color{black}H_2(\omega; g=x_1, x_0, u_0)}$ and ${\color{black}H_1(\omega; g=x_2, x_0, u_0)}$ are visualized as the Bode plots of gain and phase in Figure~\ref{fig:bode}. 
The response ${\color{black}H_2(\omega; g=x_1, x_0, u_0)}$ is due to the presence of the nonlinear term $x_2^2$ and behaves as the 3rd-order lag element. 
The response ${\color{black}H_1(\omega; g=x_2, x_0, u_0)}$ corresponds to that of the 1st-order lag element. 

\begin{figure}[t]
\centering
\includegraphics[width=0.48\textwidth]{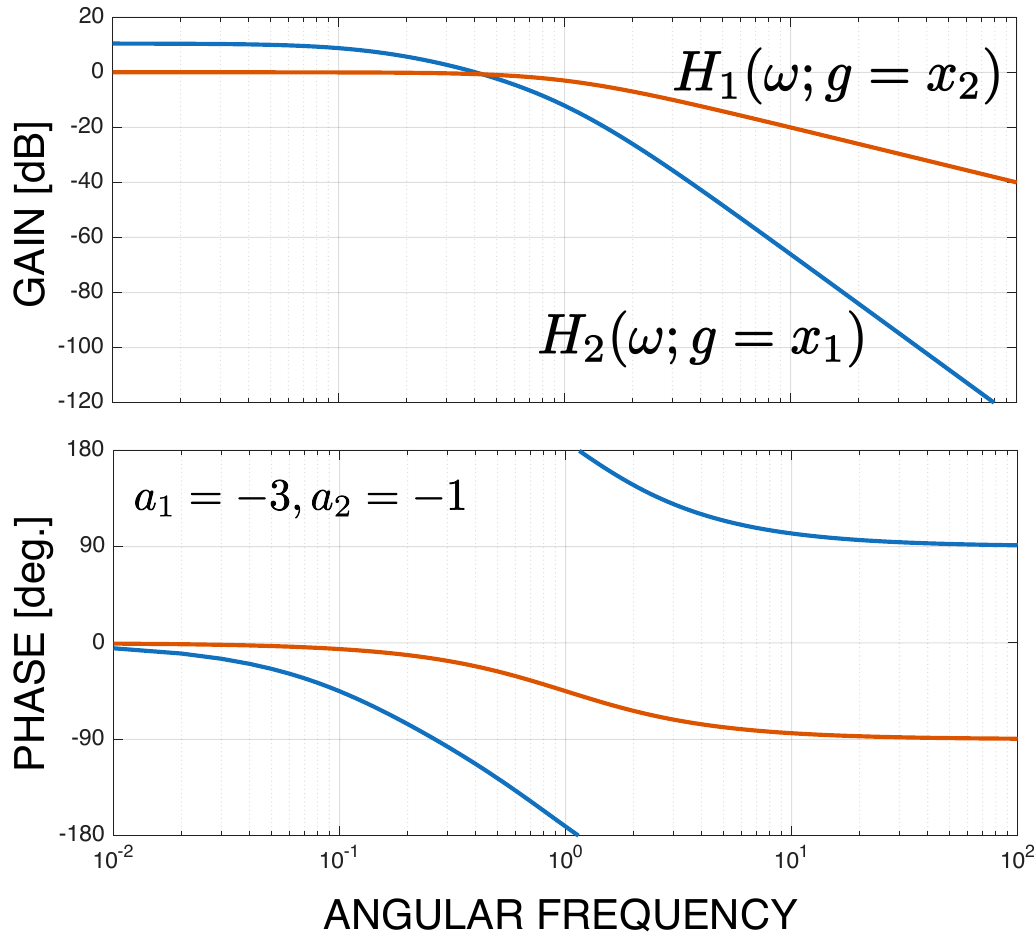}
\caption{Bode plots for frequency responses of the 2d nonlinear example: \emph{orange} for $\color{black}H_1(\omega; g=x_2, x_0, u_0)$ and \emph{blue} for $\color{black}H_2(\omega; g=x_1, x_0, u_0)$.}
\label{fig:bode}
\end{figure}

Here, we can recover the response ${\color{black}H_2(\omega; g=x_1, x_0, u_0)}$ by the Laplace transform of a lifted linear system.  
By introducing the new six states $z_1:=x_1, z_2:=x_2, z_3:=u, z_4:=x_2^2, z_5:=x_2u$, and $z_6:=u^2$, it is possible to lift the skew-product form of the nonlinear system \eqref{eq:2d_nonlinear} to the following linear system:
\begin{equation}
\begin{bmatrix}
\dot{z}_1 \\ \dot{z}_2 \\ \dot{z}_3 \\ \dot{z}_4 \\ \dot{z}_5 \\ \dot{z}_6
\end{bmatrix}
= 
\begin{bmatrix}
a_1 & 0 & 0 & 1 & 0 & 0 \\
0 & a_2 & 1 & 0 & 0 & 0 \\
0 & 0 & \ii\omega & 0 & 0 & 0 \\
0 & 0 & 0 & 2a_2 & 2 & 0 \\
0 & 0 & 0 & 0 & a_2+\ii\omega & 1 \\
0 & 0 & 0 & 0 & 0 & \ii 2\omega 
\end{bmatrix}
\begin{bmatrix}
z_1 \\ z_2 \\ z_3 \\ z_4 \\ z_5 \\ z_6
\end{bmatrix}.
\label{eq:2d_nonlinear_lifted}
\end{equation}
In this linear system, the transfer function $\hat{g}_{y(u^2)}(s)$ from the Laplace transform $\hat{z}_6$ of $z_6=u^2$ to the Laplace transform $\hat{z}_1$ of $z_1=x_1=y$ is derived as follows:
\[
\hat{g}_{y(u^2)}(s)=\frac{2}{(s-a_1)(s-2a_2)(s-(a_2+\ii\omega))}.
\] 
Then, we have $\hat{g}_{y(u^2)}(\ii 2\omega)={\color{black}H_2(\omega; g=x_1, x_0, u_0)}$. 
In addition, $\hat{g}_{yu}(s)=0$ holds from the linear system, which is consistent with ${\color{black}H_1(\omega; g=x_1, x_0, u_0)}=0$.


\section{Spectral characterizations}
\label{sec:spectral}

We try to answer a fundamental question: when can we apply Theorems~\ref{thm:FR} and \ref{thm:FR_subharmonics} to the analysis of frequency response? 
This is related to finding a sufficient condition where the Koopman resolvent has a pole of order 1 at the target angular frequency $n\omega$ or $\omega/n$. 
Here, we discuss three cases where the theorems can be applied: LTI dynamics, globally stable dynamics, and ergodic dynamics. 
The latter two dynamics are related to the classes where the Laplace transforms of the outputs are completely characterized by spectral expansions of the Koopman resolvents, which were reported in \cite{susuki2021koopmana}.

\subsection{LTI dynamics}
\label{subsec:linear}

Consider the linear case $F(x)=Ax+bu$ and $g(x,u)=c^\top x$ with the constant matrix $A\in\mathbb{R}^{d\times d}$ and vectors $b,c\in\mathbb{R}^d$. 
The symbol $\top$ stands for the transpose operation of vectors. 
Assume that $A$ has distinct $d$ eigenvalues $\lambda_1,\ldots,\lambda_d$ with negative real parts, where we can avoid the well-known condition of external resonance.  
Then, the skew-product form \eqref{eq:skew1} of the linear case has eigenvalues $\lambda_1,\ldots,\lambda_d$ and $\ii\omega$, as seen in the one-dimensional example, which are contained in the point spectrum $\sigma\sub{p}(\G\sub{forced})$ of the Koopman generator $\G\sub{forced}$. 
Then, under a suitable choice of the observable space $\setF$, for example, the space of analytic functions discussed later, we see that the Koopman eigenvalue $\ii\omega$ is simple and isolated, so that the theorems can be applied.
In fact, by finding the left and right eigenvectors of the eigenvalue $\ii\omega$ for the skew-product form as 
\begin{align}
[0,\ldots,0,1]^\top
\begin{bmatrix}
A & b \\ 0 & \ii\omega
\end{bmatrix}
&= \ii\omega\,
[0,\ldots,0,1]^\top, \nonumber\\
\begin{bmatrix}
A & b \\ 0 & \ii\omega
\end{bmatrix}
\begin{bmatrix} r_{\ii\omega} \\ 1 \end{bmatrix}
&= \ii\omega 
\begin{bmatrix} r_{\ii\omega} \\ 1 \end{bmatrix}, 
\quad r_{\ii\omega}\in\mathbb{C}^d,
\label{eq:right}
\end{align}
we can calculate the residue of the pole $\ii\omega$ for the Laplace transform $\displaystyle [c^\top, 0]\left[sI-\begin{bmatrix} A & b \\ 0 & \ii\omega \end{bmatrix}\right]^{-1}\begin{bmatrix} x_0 \\ u_0 \end{bmatrix}$ of the output as $u_0(c^\top r_{\ii\omega})$ and then derive the response $\color{black}H_1(\omega; g=c^\top x, x_0, u_0)$ as follows:
\[
{\color{black}H_1(\omega; g=c^\top x, x_0, u_0)}=c^\top r_{\ii\omega}.
\]
By using \eqref{eq:right} and the assumption $\ii\omega\notin\sigma\sub{p}(A)$ (spectrum of the matrix $A$), we can compute $r_{\ii\omega}=(\ii\omega I-A)^{-1} b$ and we have
\[
{\color{black}H_1(\omega; g=c^\top x, x_0, u_0)} = c^\top (\ii\omega I-A)^{-1} b.
\]
This shows that the frequency response derived in this paper is a generalization of the classical approach to LTI systems. 

\subsection{Globally stable dynamics}
\label{subsec:stable}

From a phenomenological point of view, we need to know that the nonlinear plant \eqref{eq:plant} with the input $u(t)=u_0\ee^{\ii\omega t}$ admits a stable periodic solution with an angular frequency $n\omega$ or $\omega/n$. 
The existence has been studied historically using various tools in systems theory, including the describing function \cite{khalil2002nonlinear}, contraction analysis \cite{lohmiller1998contraction,russo2010global}, and convergence analysis \cite{pavlov2007frequency}, among others; see Appendix~\ref{appe:existence}. 
If there exists such a periodic solution, that is to say, if there exists a stable limit cycle for the skew-product form \eqref{eq:skew1}, then, by making the additional assumption of analyticity of the vector fields and observables, it becomes possible to guarantee the existence of the pole for $\R(s; \G\sub{forced})$. 
As a global version, we present the following lemma.
\begin{lemma}
Consider the nonlinear plant \eqref{eq:plant} with the input $u(t)=u_0\ee^{\ii\omega t}$. 
Assume that there exists a globally stable periodic solution with the angular frequency $n\omega$. 
If the function $F(x,u)$ is analytic, and the observable space $\setF$ is a space of analytic functions, then $\ii n\omega$ is a pole of order 1 of the Koopman resolvent $\R(s; \G\sub{forced})$ for the skew-product form \eqref{eq:skew1}.
\end{lemma}
\begin{proof}
The existence of a globally stable periodic solution implies that the associated globally stable limit cycle exists for the skew-product form \eqref{eq:skew1}. 
Then, by applying \cite[Theorem~8.1]{mezic2020spectrum} and the argument therein to the autonomous system with the globally stable limit cycle, it is shown that $\ii n\omega$ is a simple and isolated eigenvalue of the Koopman generator $\G\sub{forced}$.  
Hence, the conclusion holds. 
\end{proof}

\subsection{Ergodic dynamics}

For fixed $\omega$, the dynamics of \eqref{eq:skew1} are assumed to be on an attractor (e.g. periodic, quasi-periodic, or chaotic attactor). In the case of a compact ergodic attractor $\mathbb{A}$, the spectral expansion of the Koopman resolvent $\R(s; \G\sub{forced})$ in the $\setL^2$ space yields (see \cite{susuki2021koopmana})
\begin{equation}
\R(s; \G\sub{forced})g 
= \sum^\infty_{j=1}\frac{\P_{\ii\omega_j}g}{s-\ii\omega_{j}}
+ \int_\mathbb{R}\frac{\E(\dd\varphi)}{s-\ii\varphi}g, \quad 
\forall g\in \setL^2(\mathbb{A}),
\label{eq:expansion}
\end{equation}
where the projection operator $\P_{\ii\omega_j}$ is associated with $\ii\omega_{j}$, and $\E(\varphi)$ is a continuous spectral measure. 
The complex values $\ii\omega_{j}$ of the poles are the Koopman eigenvalues and simple from the ergodicity assumption. 
Lemma~\ref{lemma:point} implies that $\omega_j$'s contain $n\omega$.  
Here, let us assume that the eigenvalue $\ii n\omega$ is isolated. 
Then, it is a pole of order 1 for $\R(s; \G\sub{forced})$, so that the theorems can be applied. 
Thus, using Theorem~\ref{thm:FR} with \eqref{eq:expansion}, we can derive the frequency response as
\[
\begin{aligned}
{\color{black}H_n(\omega; g, x_0, u_0)} & \\
&\makebox[-15mm]{}= u_0^{-n}
\lim_{s\to\ii n\omega} (s-\ii n\omega)[\R(s; \G\sub{forced})g](x_0,u_0) \\
&\makebox[-15mm]{}= u_0^{-n}(\P_{\ii n\omega}g)(x_0,u_0).
\end{aligned}
\]
The subharmonic case for the eigenvalue $\ii\omega/n$ also holds in the same argument above. 
\begin{remark}
If the eigenvalue $\ii n\omega$ 
is not isolated (see the ergodic flow on the torus), then we are not able to directly apply the theorems to derive the frequency response. 
However, for the ergodic case, because the Koopman generator is a unitary operator acting on the $\setL^2$ space, the eigenprojection can be directly derived with the inner product equipped with the $\setL^2$ space. 
In this case, without assuming an isolated eigenvalue, it becomes possible to derive the frequency response. 
\end{remark}
\begin{remark}
Although beyond the scope of this paper, several extensions of the frequency response might be interesting for systems with point spectra only, mixed spectra, and continuous spectra, such as the weakly mixing case. 
The involvement of continuous spectrum in the theory of frequency response of nonlinear systems is our future direction. 
The ergodicity assumption excludes the case that two $\lambda_j$ are identical to $\ii\omega$, that is, the well-known condition of the external resonance phenomenon. 
\end{remark}

\section{Concluding remarks}
\label{sec:outro}

This paper proposed a novel formulation of the frequency response for the nonlinear systems. 
By utilizing the theory of Koopman resolvent presented in \cite{susuki2021koopmana}, we showed that the frequency response is derived through the Laplace transform of the output of a target nonlinear plant. 
The response is a complex-valued function of the driving angular frequency, allowing one to draw the so-called Bode plots, which display gain and phase characteristics. 
Sufficient conditions for the existence of frequency response were presented for the three classes of dynamics.  

We envision several perspectives for future research.
One is to analyze well-posedness, continuity, and differentiability of the frequency responses ${\color{black}H_n(\omega; g, x_0, u_0)}$ and ${\color{black}H_{1/n}(\omega; g, x_0, u_0)}$. 
This will be conducted with the Koopman theory for parametrized dynamical systems \cite{katayama2025}. 
Another perspective is to utilize ${\color{black}H_n(\omega; g, x_0, u_0)}$ and ${\color{black}H_{1/n}(\omega; g, x_0, u_0)}$ for nonlinear systems analysis, such as stability and passivity. 
Finally, an input-output relation or transfer function could be defined for the nonlinear plant.

\appendices

\section{Calculation of Koopman mode decompositions}
\label{appe:KMD}

The Koopman mode decompositions of $x_1(t)$ and $x_2(t)$ for the 2d nonlinear example are the following:
\[
\begin{aligned}
x_1(t) &= \ee^{a_1t}\phi_{a_1}(x(0),u(0)) \\
& +\ee^{2a_2t}\phi_{a_2}(x(0),u(0))^2\frac{-1}{a_1-2a_2} \\
& +\ee^{(a_2+\ii\omega)t}\phi_{a_2}(x(0),u(0))\phi_{\ii\omega}(x(0),u(0))\cdot \\
& \makebox[10em]{}\cdot \frac{2}{(\ii\omega-a_1+a_2)(\ii\omega-a_2)} \\
& +\ee^{\ii 2\omega t}\phi_{\ii\omega}(x(0),u(0))^2{\color{black}H_2(\omega; g=x_1, x_0, u_0)},\\
x_2(t) &= \ee^{a_2t}\phi_{a_2}(x(0),u(0)) \\
& +\ee^{\ii\omega t}\phi_{\ii\omega}(x(0),u(0)){\color{black}H_1(\omega; g=x_2, x_0, u_0)}.
\end{aligned}
\]

\section{Existence of globally stable periodic solutions}
\label{appe:existence}

\begin{lemma}
The nonlinear plant \eqref{eq:plant} with the input $u(t)=u_0\ee^{\ii\omega t}$ has a globally stable periodic solution with the angular frequency $\omega$ if one of the following conditions is satisfied, as appropriate:
\begin{itemize}
\item[(i)] The system $\dot{x}=F_0(x)(=F(x,u=0))$ is a time-invariant contracting system in the definition of \cite{lohmiller1998contraction}; 
\item[(ii)] The system $\dot{x}=F(x,u)$ is regular and uniformly convergent in the uniformly bounded steady-state property for the class of bounded continuous inputs, as in the definition of \cite{pavlov2007frequency}.
\end{itemize}
\end{lemma}
\begin{proof}
For (i), see Section~3.7 (vi) of \cite{lohmiller1998contraction} or Theorem~2 of \cite{russo2010global};  
for (ii), see Theorem~2 of \cite{pavlov2007frequency}.
\end{proof}

\end{document}